\documentclass[a4paper,10pt,english]{siamltex}
\usepackage{amsmath,amsfonts,amssymb,latexsym}
\usepackage{graphicx}
\usepackage[latin1]{inputenc}
\usepackage{color}
\usepackage{pstricks}
\usepackage[latin1]{inputenc}

\newcommand{\mat}[1]{{\testgreek#1\ifgreek\boldsymbol#1\else
                      \mathbf#1\fi}} 

\newif\ifgreek
\def\testgreek#1{
  \ifx#1\alpha\greektrue\else
  \ifx#1\beta\greektrue\else
  \ifx#1\gamma\greektrue\else\ifx#1\Gamma\greektrue\else
  \ifx#1\delta\greektrue\else\ifx#1\Delta\greektrue\else
  \ifx#1\epsilon\greektrue\else
  \ifx#1\zeta\greektrue\else
  \ifx#1\eta\greektrue\else
  \ifx#1\theta\greektrue\else\ifx#1\Theta\greektrue\else
  \ifx#1\iota\greektrue\else
  \ifx#1\kappa\greektrue\else
  \ifx#1\lambda\greektrue\else\ifx#1\Lambda\greektrue\else
  \ifx#1\mu\greektrue\else
  \ifx#1\nu\greektrue\else
  \ifx#1\xi\greektrue\else\ifx#1\Xi\greektrue\else
  \ifx#1\pi\greektrue\else\ifx#1\Pi\greektrue\else
  \ifx#1\rho\greektrue\else
  \ifx#1\sigma\greektrue\else\ifx#1\Sigma\greektrue\else
  \ifx#1\tau\greektrue\else
  \ifx#1\upsilon\greektrue\else\ifx#1\Upsilon\greektrue\else
  \ifx#1\phi\greektrue\else\ifx#1\Phi\greektrue\else
  \ifx#1\chi\greektrue\else
  \ifx#1\psi\greektrue\else\ifx#1\Psi\greektrue\else
  \ifx#1\omega\greektrue\else\ifx#1\Omega\greektrue\else
  \ifx#1\varepsilon\greektrue\else
  \ifx#1\vartheta\greektrue\else
  \ifx#1\varrho\greektrue\else
  \ifx#1\varsigma\greektrue\else
  \ifx#1\varphi\greektrue\else
     \greekfalse
  \fi\fi\fi\fi\fi\fi\fi\fi\fi\fi
  \fi\fi\fi\fi\fi\fi\fi\fi\fi\fi
  \fi\fi\fi\fi\fi\fi\fi\fi\fi\fi
  \fi\fi\fi\fi\fi\fi\fi\fi\fi}

\providecommand*{\Z}{\mathbb{Z}}      
\newcommand{\N}{\mathbb{N}}

\newcommand{\TT}{\mathbb{T}^{2}} 
\newcommand{\RR}{\mathbb{R}^{2}}

\providecommand*{\mrm}[1]{\mathrm{#1}}
\newcommand{\iu}{\mrm{i}}
\newcommand{\eu}{\mrm{e}}
\newcommand{\spkx}{\mrm{i}k\cdot x}
\newcommand{\spnx}{\mrm{i}n\cdot x}
\newcommand{\spmx}{\mrm{i}m\cdot x}
\newcommand{\diff}{\,\mrm{dx}}

\newcommand{\blk}[1]{\mathbf{#1}}

\newcommand{\qed}{\nobreak \ifvmode \relax \else
      \ifdim\lastskip<1.5em \hskip-\lastskip
      \hskip1.5em plus0em minus0.5em \fi \nobreak
      \vrule height0.75em width0.5em depth0.25em\fi}
 
\title{On the Spectrum of an Operator Pencil with Applications to Wave Propagation in Periodic and Frequency Dependent Materials}
\author{Christian Engstr\"om and Markus Richter}

\begin{document}

\maketitle

\begin{abstract}
We study wave propagation in periodic and frequency dependent materials. The approach in this paper leads to spectral analysis of a quadratic operator pencil where the spectral parameter relates to the quasimomentum and the frequency is a parameter. We show that the underlying operator has a discrete spectrum, where the eigenvalues are symmetrically placed with respect to the real and imaginary axis. Moreover, we discretize the operator pencil with finite elements and use a Krylov space method to compute eigenvalues of the resulting large sparse matrix pencil.
\end{abstract}

\section{Introduction}

In this paper, we study wave propagation in a two-dimensional periodic medium, where the material components depend on the frequency. The underlying operator has a band-structure, which by appropriate choice of the material configuration has gaps \cite{Joannopoulos1995}. Such structures have numerous applications in optics, photonics and microwave engineering \cite{Joannopoulos1995}. Electromagnetic wave propagation is in general frequency dependent, since dispersive effects such as resonances or relaxation processes are present in the material \cite{Cessenat1996,Jackson1999}. The frequency dependence of the material is in most cases disregarded in spectral analysis. This is justified in some cases where the frequency dependence is weak. However, the frequency dispersion can in general not be ignored. 

In the case of frequency independent materials, considerable mathematical prog--ress has been made \cite{Reed+Simon1978,Kuchment1993}. In the frequency dependent case, the spectral problem is non-linear, which complicates the analysis. There exist two different approaches to handle the frequency dependence. Within the first method, a real-valued quasimomentum  (the Floquet-Bloch wave vector) is specified and the spectral parameter is related to the time frequency \cite{Tip+Moroz+Combes2000,Spence+Poulton2005}. The second option is to study the quasimomentum $k$ as a function of a real frequency $\omega$. The non-linearity is arbitrary within the first approach, but the second approach leads regardless of the frequency dependence in the material model to a quadratic non-linearity. The quadratic non-linearity simplifies the spectral analysis and the numerical calculations. However, only numerical simulations exist for this approach \cite{Smith+Mailhiot1986,Hermann+etal2008}. 

We use the second option to study wave propagation in a periodic and frequency dependent medium that for a frequency interval is characterized by a space- and frequency-dependent real-valued permittivity $\epsilon (x,\omega)$. The waves propagate in a non-magnetic material with the permittivity $\epsilon (x_{1},x_{2},\omega)$ independent of the third coordinate $x_{3}$. Let $E$ and $H$ denote the electric and magnetic waves, respectively. The $x_{3}$- independent electromagnetic wave $(E,H)$ is decomposed into transverse electric (TE) polarized waves $(E_{1},E_{2},0,0,0,H_{3})$ and transverse magnetic (TM) polarized waves $(0,0,E_{3},H_{1},H_{2},0)$ \cite{Kuchment2001}. This decomposition reduces the spectral problem for the Maxwell operator to one scalar equation for $H_{3}$ and one scalar equation for $E_{3}$. We consider TM polarized waves. This gives a Helmholtz type of equation in $E_{3}$, but the analysis also applies to the TE case.

The approach in this paper leads to the spectral analysis of a quadratic operator pencil, where the spectral parameter relates to the quasimomentum and where the frequency is a parameter. We show that the underlying operator has a discrete spectrum, where the eigenvalues are symmetrically placed with respect to the real and imaginary axis (Hamiltonian symmetry). A finite element method is used to discretize the operator pencil and the resulting matrix pencil is transformed into a gyroscopic pencil. We implement the structure preserving skew-Hamiltonian isotropic implicitly restarted Arnoldi algorithm SHIRA \cite{mehrmann+watkins2001} to compute eigenvalues of the gyroscopic matrix pencil.

\section{Bloch solutions and spectral problems}

Let $\Gamma$ denote the lattice $2\pi\mathbb{Z}^{2}$ and $\epsilon$ a measurable function with
\begin{equation}
	\epsilon (x+\gamma,\omega)=\epsilon (x,\omega),\quad \forall \gamma \in \Gamma
	\label{periodic}
\end{equation}
for almost all $x=(x_{1},x_{2})\in  \mathbb{R}^{2}$. A function with the property (\ref{periodic}) is called $\Gamma$-periodic. The fundamental periodicity domain of the lattice $\Gamma$ is $\Omega=(-\pi,\pi]^{2}$ and $|\Omega|=4\pi^{2}$ denotes the area of the domain. In physics, a fundamental domain is often called a \emph{Wigner-Seitz cell} \cite{Kittel1986,Joannopoulos1995}. We assume that $\epsilon$ is in $L^{\infty}(\mathbb{R}^{2})$ for each $\omega\in [\omega_{0},\omega_{1}]$ and that $\epsilon$ is invertible, positive and bounded
\begin{equation}
	0<c_{0}\leq \epsilon (x,\omega) \leq c_{1},
\end{equation}
for almost all $x\in  \mathbb{R}^{2}$ and $\omega\in [\omega_{0},\omega_{1}]$. Let $\mathbb{T}^{2}=\mathbb{R}^{2}/\Gamma$ denote the torus in two dimensions and $C^{\infty}(\mathbb{T}^{2})$ the space of smooth complex-valued $\Gamma$-periodic functions. The space $L^{2}(\mathbb{T}^{2})$ is defined to be the completion of $C^{\infty}(\mathbb{T}^{2})$ in the $L^{2}$-norm and $H^{1}(\mathbb{T}^{2})$ is the completion of $C^{\infty}(\mathbb{T}^{2})$ in the $H^{1}$-norm. 

Let $\nabla$ denote the gradient with respect to $x\in\mathbb{R}^{2}$ and let $\nabla\cdot$ denote the divergence. The transverse magnetic polarized waves $(0,0,E_{3},H_{1},H_{2},0)$ can be written in terms of the electric field $E$ alone. The equation
\begin{equation}
	-\Delta v=\omega^{2} \epsilon(x,\omega)v
	\label{eq:bas}
\end{equation} 
models a time-harmonic monochromatic electromagnetic wave with frequency $\omega$ and electric polarization $E(x)=(0,0,v(x))$. A Bloch solution of (\ref{eq:bas}) is a non-zero solution of the form
\begin{equation}
	v(x)=\eu^{\spkx}u(x),
	\label{Bloch}
\end{equation}
where $u$ is a $\Gamma$-periodic function and $k\in\mathbb{C}^{2}$ is the quasimomentum (or the Floquet-Bloch wave vector) \cite{Kuchment1993}. Since 	
$\nabla (\eu^{\spkx}u(x))=\eu^{\spkx}(\nabla+\iu k)u(x)$ the Bloch solutions of (\ref{eq:bas}) are  $\Gamma$-periodic solutions of
\begin{equation} 
-(\nabla+\iu k)\cdot (\nabla+\iu k)u(x)=\omega^{2}\epsilon(x,\omega)u.
\label{shifted}
\end{equation}
The frequency $\omega$ as a multi-valued function of the quasimomentum $k$ is called the dispersion relation and the graph of the dispersion relation defines the Bloch variety \cite{Kuchment1993}.
\begin{definition}
The complex Bloch variety is the set
\begin{equation}
	B=\{(k,\omega)\in\mathbb{C}^{2}\times\mathbb{C}\,|\text{the problem (\ref{eq:bas}) has a non-zero Bloch solution}\}
	\label{BV}
\end{equation}
\end{definition}

Analytic properties of the Bloch variety $B$ are discussed in Kuchment \cite{Kuchment1993}. The Bloch solutions of (\ref{eq:bas}) can be determined from spectral problems. In the frequency independent case $\epsilon=\epsilon (x)$, Bloch solutions with $k \in\mathbb{R}^{2}$ are given by eigenvectors of a selfadjoint operator. In the frequency dependent case $\epsilon=\epsilon (x,\omega)$, the spectral problem is non-linear, which complicates the analysis. We will consider solutions $(k,\omega)\in\mathbb{C}^{2}\times\mathbb{R}$ in the frequency dependent case and use the real frequency $\omega$ as a parameter. The two cases are discussed below.

\subsection{The frequency independent case}
\label{non-disp}

Assume that $\epsilon=\epsilon (x)$ independent of the frequency $\omega$ and define the operator
\begin{equation}
	A(x,\nabla )v=-\frac{1}{\epsilon(x)}\Delta v.
	\label{Ainf}
\end{equation} 
The spectrum of elliptic operators with periodic coefficient has been studied in detail by many authors including Odeh and Keller \cite{Odeh+Keller1964}, Reed and Simon \cite{Reed+Simon1978}, and Figotin and Kuchment \cite{Figotin+Kuchment1996a}. The $L_{2}$-spectrum $\sigma (A)$ has a band structure and the existence of band-gaps for certain geometrical structures has been proved \cite{Figotin+Kuchment1996a, Figotin+Kuchment1996b,Figotin+Kuchment1998}. From the Floquet-Bloch theory \cite{Reed+Simon1978,Kuchment1993} follows the spectrum of $A$ in the infinite periodic medium can be obtained from the $k$ - dependent family of $\Gamma$-periodic solutions (\ref{shifted}). We give below the well-known construction of the band structure. By definition, the function $u\in H^{1}(\TT)$ is a solution of the shifted problem
\begin{equation} 
	A(k)u=\omega^{2}u,\quad A(k)u=-\frac{1}{\epsilon(x)}(\nabla+\iu k)\cdot (\nabla+\iu k)u(x)
\label{shifted2}
\end{equation}
if the integral identity
\begin{equation}
	\int_{\TT}(\nabla+\iu k)u\cdot\overline{(\nabla+\iu k)\phi}\diff=\omega^{2}	\int_{\TT}u\bar{\phi}\epsilon\diff
\end{equation}
is satisfied for all $\phi\in H^{1}(\TT)$. Each operator $A(k)$ is selfadjoint in $L^{2}(\mathbb{T}^{2})$, with scalar product weighted by $\epsilon (x)$ and $k\in\mathbb{R}^{2}$. Since $A(k)$ is non-negative and has a compact resolvent, the spectrum of $A(k)$ consists of eigenvalues 
\[
0\leq \omega_{0}^{2}(k)\leq \omega_{1}^{2}(k)\leq ...
\]
where the spectral parameter $\omega^{2}$ represent the square of the time frequency $\omega$ of the wave. 
The main spectral result is that the spectrum of $A$ (\ref{Ainf}) is the union
\begin{equation}
	\sigma (A)=\bigcup_{n\geq 0} S_{n},\quad S_{n}=[\min_{k}\omega_{n}^{2}(k),\max_{k}\omega_{n}^{2}(k)].
	\label{spA}
\end{equation}
The intervals $S_{n}$ are called the stability zones of the operator $A$ and the complement of $\sigma (A)$ is called the instability zone. The successive segments $S_{n}$ and $S_{n+1}$ may overlap,  but we have a gap in the spectrum $\sigma (A)$ if the intersection $S_{n}\cap S_{n+1}$ is empty for some $n$. That is, if $\omega^{2}$ belongs to the instability zone for all $k\in\mathbb{R}^{2}$. The spectrum $\sigma (A)$ can also be described in terms of the Bloch variety (\ref{BV}). The real Bloch variety is the intersection
\begin{equation}
	B_{\mrm{real}}=B \cap \mathbb{R}^{3}
	\label{realBV}
\end{equation}
and the spectrum of $A$ is the projection of the real Bloch variety onto the $\omega$- axis. This alternative description of the spectrum is important in the frequency dependent case below.

\section{The frequency dependent case}

We assumed above that the permittivity $\epsilon$ is independent of the frequency $\omega$. The spectrum $\sigma (A)$ coincides in the frequency independent case with the range of the dispersion relation $\omega^{2}(k)$. The permittivity can be constant in a frequency interval, but only vacuum is independent of the frequency. If we exclude vacuum, the Kramers-Kronig relations \cite{Cessenat1996,Jackson1999} implies that the imaginary part of $\epsilon (\omega)$ cannot be zero for all $\omega$. We focus on the dispersive case $\epsilon=\epsilon (x,\omega)$, where $\epsilon$ is real-valued in the frequency interval $[\omega_{0},\omega_{1}]$. That is, we assume that there is no absorption line in the frequency interval under consideration. This is a reasonable model for many dialectic materials (insulators) for a range of frequencies in the micro-wave regime and in the optical regime \cite{Jackson1999}. Below, we study the quadratic pencil in the case of a real-valued permittivity function. 

\subsection{The quadratic operator pencil}

We presented in Section \ref{non-disp} results for the shifted operator (\ref{shifted}) in the case $\epsilon=\epsilon (x)$, where the spectral parameter $\omega^{2}$ corresponds to the square of the time frequency $\omega$. We consider below the frequency dependent case and relate the spectrum parameter to the quasimomentum $k$. We reduce the spectral problem in $k$ to a problem in the complex amplitude of the vector $k$. 

Let $k=\lambda\hat{k}$, where $\lambda\in\mathbb{C}$ and $\hat{k}$ is a unit vector in $\RR$. Define a family of spectral problems as
\begin{equation} 
-(\nabla+\iu\lambda\hat{k})\cdot (\nabla+\iu\lambda\hat{k})u(x)=\omega^{2}\epsilon(x,\omega)u,
\label{classA}
\end{equation}
where $u$ and $\epsilon$ are $\Gamma-$ periodic functions. This spectral problem in $(\lambda,u)$ corresponds in the frequency independent case to the problem $A(k)u=\omega^{2}u$ in (\ref{shifted2}), but we consider $k(\omega)$ and not $\omega(k)$. The family of spectral problems (\ref{classA}) is quadratic in the spectral parameter $\lambda$, with parameter $\omega$. We write below the classical problem (\ref{classA}) in a weak form. Define the sesquilinear forms
\begin{equation}
	 \begin{aligned}
			a_{0} &: H^{1}(\mathbb{T}^{2})\times H^{1}(\mathbb{T}^{2}) & \rightarrow\mathbb{C}\\
			a_{1} &: L^{2}(\mathbb{T}^{2})\times H^{1}(\mathbb{T}^{2}) & \rightarrow\mathbb{C}\\
			a_{2} &: L^{2}(\mathbb{T}^{2})\times L^{2}(\mathbb{T}^{2}) & \rightarrow\mathbb{C}
	\end{aligned}
\end{equation}
where

\begin{equation}
	 \begin{aligned}
		a_{0}(u,v) & =\int_{\mathbb{T}^{2}}\nabla u\cdot\nabla \overline{v}-\omega^{2}\epsilon u\overline{v}\diff,\\
		a_{1}(u,v) & =2\iu\int_{\mathbb{T}^{2}}u\hat{k}\cdot\nabla\overline{v}\diff,\\
			a_{2}(u,v)& =\int_{\mathbb{T}^{2}}u\overline{v}\diff.
	\end{aligned}
\end{equation}
The  weak solutions of (\ref{classA}) are defined as eigenvectors $u$ of $H^{1}(\mathbb{T}^{2})\backslash\{0\}$ and eigenvalues $\lambda\in \mathbb{C}$ which satisfy
\begin{equation}  
	\lambda^{2} a_{2}(u,v)+\lambda a_{1}(u,v)+a_{0}(u,v)=0,
	\label{ii}
\end{equation}
for all $v\in H^{1}(\mathbb{T}^{2})$. 

\begin{lemma}
The sesquilinear forms $a_{0}$, $a_{1}$ and $a_{2}$ are bounded for all elements $u,v\in H^{1}(\mathbb{T}^{2})$.
\label{bounded}
\end{lemma}
\begin{proof}
The sesquilinear forms $a_{0}$ and $a_{2}$ are trivially bounded and the boundedness of $a_{1}$ follows from the inequality $|ab|\leq \delta |a|^{2}/2+|b|^{2}/(2\delta),\, \delta>0$.
\end{proof}

Let $(\cdot,\cdot )_{H^{1}(\mathbb{T}^{2})}$ denote the scalar product in $H^{1}(\mathbb{T}^{2})$. The spectral problem (\ref{ii}) can be written as
\begin{equation}
		\lambda^{2}(A_{2}u,v)_{H^{1}(\mathbb{T}^{2})}+\lambda (A_{1}u,v)_{H^{1}(\mathbb{T}^{2})}+(A_{0}u,v)_{H^{1}(\mathbb{T}^{2})}=0,
		 \label{OP}
\end{equation}
where the operators $A_{0}$, $A_{1}$ and $A_{2}$ are defined by the bounded sesquilinear forms
\begin{equation}
	(A_{k}u,v)_{H^{1}(\mathbb{T}^{2})}=a_{k}(u,v),\quad k=0,1,2.
\end{equation}
According to the Riesz theorem these operators are on $H^{1}(\mathbb{T}^{2})$ unique, linear and bounded \cite{Kato1980}.
The equality (\ref{OP}) that holds for all $v\in H^{1}(\mathbb{T}^{2})$ is equivalent to the operator equation
\begin{equation}
	A_{0}u+\lambda A_{1}u+\lambda^{2}A_{2}u=0
		 \label{OP2}
\end{equation}
in $H^{1}(\mathbb{T}^{2})$.

Let $\mathcal{L}(H^{1}(\mathbb{T}^{2}))$ denote the set of all bounded linear operators on $H^{1}(\mathbb{T}^{2})$. We introduce the operator pencil $Q:\mathbb{C}\rightarrow \mathcal{L}(H^{1}(\mathbb{T}^{2}))$
\begin{equation}
	Q(\lambda)=A_{0}+\lambda A_{1}+\lambda^{2}A_{2},\quad \lambda\in \mathbb{C}.
	\label{pencil}
\end{equation}
The quadratic eigenvalue problem is then: Find $\lambda\in \mathbb{C}$ and a non-zero $u\in H^{1}(\mathbb{T}^{2})$ such that
\begin{equation}
	Q(\lambda)u=0.
	\label{Keq}
\end{equation}
The adjoint operator pencil is
\begin{equation}
	Q^{*}(\lambda)=A_{0}^{*}+\overline{\lambda}A_{1}^{*}+\overline{\lambda}^{2}A_{2}^{*},\quad\lambda\in \mathbb{C}
\end{equation}
and the pencil $Q(\lambda)$ is said to be selfadjoint if $A_{0}$, $A_{1}$ and $A_{2}$ are selfadjoint \cite{Markus1988}. The lemma below shows that the pencil $Q(\lambda)$ is selfadjoint.
\begin{lemma}
The operators $A_{0}$, $A_{1}$ and $A_{2}$ are self-adjoint
\label{Asa}
\end{lemma}
\noindent
\begin{proof} 
Since the operators $A_{0}$, $A_{1}$ and $A_{2}$ are defined on the full space $H^{1}(\mathbb{T}^{2})$ the lemma follows from the equalities
\[
	(A_{0}u,v)_{H^{1}(\mathbb{T}^{2})}=a_{0}(u,v)=\overline{(A_{0}v,u)}_{H^{1}(\mathbb{T}^{2})}=\overline{(v,A_{0}^{*}u)}_{H^{1}(\mathbb{T}^{2})}=(A_{0}^{*}u,v)_{H^{1}(\mathbb{T}^{2})}
\]
\[
	(A_{1}u,v)_{H^{1}(\mathbb{T}^{2})}=a_{1}(u,v)=\overline{(A_{1}v,u)}_{H^{1}(\mathbb{T}^{2})}=\overline{(v,A_{1}^{*}u)}_{H^{1}(\mathbb{T}^{2})}=(A_{1}^{*}u,v)_{H^{1}(\mathbb{T}^{2})},
\]
\[
	(A_{2}u,v)_{H^{1}(\mathbb{T}^{2})}=a_{2}(u,v)=\overline{(A_{2}v,u)}_{H^{1}(\mathbb{T}^{2})}=\overline{(v,A_{2}^{*}u)}_{H^{1}(\mathbb{T}^{2})}=(A_{2}^{*}u,v)_{H^{1}(\mathbb{T}^{2})}.
\]
\end{proof}

\subsection{Spectral properties} 

We consider in this section spectral properties of the operator pencil (\ref{pencil}). Fredholm theory is used to prove that the spectrum of $Q$ is discrete, where the eigenvalues are symmetrically placed with respect to the real and imaginary axis. Moreover, we prove that the pencil $Q(\lambda)$ cannot be reduced to a monic bounded pencil, that is, to a pencil of the form
\begin{equation}
	Q(\lambda)=B_{0}+\lambda B_{1}+\lambda^{2}\mrm{I},\quad \lambda\in \mathbb{C}.
\end{equation}
Nonmonic pencils has in general a more complex spectral structure \cite{Markus1988}.
\begin{lemma} 
The operator $A_{2}$ has an unbounded inverse.
\end{lemma}

\begin{proof} 
Since $A_{2}>0$ we have \cite{Birman+Solomjak1987}
\begin{equation}
	\begin{aligned}
		\|A_{2}u\|^{2}_{H^{1}(\mathbb{T}^{2})} & \leq \|A_{2}\|_{\mathcal{L}(H^{1}(\mathbb{T}^{2}))}(A_{2}u,u)_{H^{1}(\mathbb{T}^{2})}\\
		&=C\int_{\mathbb{T}^{2}}|u|^{2}\diff.
	\end{aligned}
	\label{A2inv}
\end{equation}
The operator $A_{2}$ has a bounded inverse if and only if there exist a positive constant $c$ such that $\|A_{2}u\|_{H^{1}(\mathbb{T}^{2})}\geq c\|u\|_{H^{1}(\mathbb{T}^{2})}$ for all $u\in H^{1}(\mathbb{T}^{2})$. The inequality (\ref{A2inv}) implies that such constant $c$ cannot exist (take for instance $u(x)=\eu^{\iu n\cdot x}$, $n\in \N$), but $A_{2}$ has an inverse since the null space $N(A_{2})$ is trivial.
\end{proof}
 
A bounded operator $T$ is Fredholm if the range $R(T)$ is closed and if the null space $N(T)$, and the complement of the range $R(T)$ are finite
dimensional. The Fredholm index of a Fredholm operator $T$ is defined as
\begin{equation}
    \mrm{ind}\, T=\mrm{nul}\,T-\mrm{def}\, T,
    \label{index}
\end{equation}
where the nullity $\mrm{nul}\, T$ is the dimension of $N(T)$ and the
deficiency $\mrm{def}\,T$ is the dimension of the complement to range
$R(T)$. The resolvent set $\rho$ is the set of all $\lambda\in \mathbb{C}$, such that the
operator $T(\lambda )$ is continuously invertible and the spectrum is defined by
the complement
\begin{equation}
    \sigma (T)=\mathbb{C}\backslash\rho (T).
\end{equation}
The essential spectrum $\sigma_{\mrm{ess}}(T)$ of $T$ is the set of all
$\lambda\in\sigma (T)$ such that $T(\lambda)$ is not a Fredholm operator. The
complement of the essential spectrum
\begin{equation}
    \sigma_{\mrm{discr}}(T)=\sigma (T)\backslash\sigma_{\mrm{ess}}(T)
\end{equation}
is called the discrete spectrum. The discrete spectrum consists of eigenvalues of finite geometrical multiplicity, $\mrm{nul}\, T<\infty$. We shall show that the spectrum of $Q(\lambda)$ is discrete. The proof consists of two steps. We prove that the operator pencil $Q(\lambda)$ is a sum of compact operators and a Fredholm operator.

\begin{lemma}
The operators $A_{1}$ and $A_{2}$ are compact in $H^{1}(\mathbb{T}^{2})$.
\end{lemma}
\begin{proof}
We show that $A_{1}$ is compact.  Let $\{u_{n}\}\subset H^{1}(\mathbb{T}^{2})$ be a given bounded sequence. Then since $H^{1}(\mathbb{T}^{2})$ is compactly embedded into $L^{2}(\mathbb{T}^{2})$ \cite{Taylor1981}, the sequence $\{u_{n}\}$ has a convergent subsequence $\{u_{m}\}$ in $L^{2}(\mathbb{T}^{2})$. Let $u_{m}$ and $u_{m'}$ denote two element in the subsequence $\{u_{m}\}$. From the boundedness of the Hermitian form $a_{1}$  follows
\begin{equation} 
    \begin{aligned}
			\|A_{1}u_{m}-A_{1}u_{m'}\|^{2}_{H^{1}(\mathbb{T}^{2})} &= |a_{1}(u_{m}-u_{m'},A_{1}u_{m}-A_{1}u_{m'})|\\
			&\leq  C\|A_{1}u_{m}-A_{1}u_{m'}\|_{H^{1}(\mathbb{T}^{2})}\|u_{m}-u_{m'}\|_{L^{2}(\mathbb{T}^{2})}.
    \end{aligned}
\end{equation}
Then since $\{u_{m}\}$ is a Cauchy sequence in $L^{2}(\mathbb{T}^{2})$, the sequence $\{A_{1}u_{m}\}$ is a Cauchy sequence in $H^{1}(\mathbb{T}^{2})$, which converges since $H^{1}(\mathbb{T}^{2})$ is complete. Using the inequality $\|u\|_{L^{2}(\mathbb{T}^{2})}\leq C\|u\|_{H^{1}(\mathbb{T}^{2})}$, the compactness of $A_{2}$ is proved in the same way. 
\end{proof}

\begin{theorem}
The operator pencil $Q(\lambda)$ is a Fredholm operator with index zero.
\label{Fredholm}
\end{theorem}
\begin{proof} 
The operator $A_{0}$ in (\ref{pencil}) can be written of the form
\begin{equation}
    (A_{0}u,v)_{H^{1}(\mathbb{T}^{2})}=(A_{0}^{(1)}u,v)_{H^{1}(\mathbb{T}^{2})}+(A_{0}^{(2)}u,v)_{H^{1}(\mathbb{T}^{2})},
\end{equation}
where $A_{0}^{(1)}$ and $A_{0}^{(2)}$ are defined by the bounded sesquilinear forms
\begin{equation}
    (A_{0}^{(1)}u,v)_{H^{1}(\mathbb{T}^{2})}=-\omega^{2}\int_{\mathbb{T}^{2}}\epsilon u\overline{v} \diff,\quad (A_{0}^{(2)}u,v)_{H^{1}(\mathbb{T}^{2})}=\int_{\mathbb{T}^{2}}\nabla u\cdot \nabla \overline{v}\diff.
\end{equation}
The operator $A_{0}^{(1)}$ is selfadjoint and compact since $\epsilon\in L^{\infty}(\mathbb{T}^{2})$ is real-valued and $H^{1}(\mathbb{T}^{2})$ is compactly embedded into $L^{2}(\mathbb{T}^{2})$. 
Represent $u,v\in H^{1}(\mathbb{T}^{2})$ with its Fourier series:
\begin{equation}
	u(x)=\sum_{n\in\mathbb{Z}^{2}}\hat{u}_{n}\eu^{\spnx},\quad 	v(x)=\sum_{m\in\mathbb{Z}^{2}}\hat{v}_{m}\eu^{\spmx}.
\end{equation}
Since the gradients of $u$ and $v$ are in $L^{2}(\mathbb{T}^{2})$ the Plancherel's identity \cite{Yosida1980} gives
\begin{equation}
     \begin{aligned}
  			 (A_{0}^{(2)}u,v)_{H^{1}(\mathbb{T}^{2})} &=\int_{\mathbb{T}^{2}} \sum_{n\in\mathbb{Z}^{2}}\iu n\hat{u}_{n}\eu^{\spnx}\cdot \sum_{m\in\mathbb{Z}^{2}}-\iu m\hat{v}_{m}^{*}\eu^{-\iu m\cdot x}\diff\\
      &=|\Omega|\sum_{n\in\mathbb{Z}^{2}}|n|^{2}\hat{u}_{n}\hat{v}^{*}_{n},
     \end{aligned}
\end{equation}
which implies that the dimension of the null space is one. The dimension of the orthogonal complement of the range is also one, since $R(T)^{\bot}=N(T)$ for self-adjoint operators. That is, $A_{0}^{(2)}$ is a Fredholm operator with index (\ref{index}) zero. The self-adjoint operator pencil $Q(\lambda )$ is a sum of $A_{0}^{(2)}$ and compact operators, which is a Fredholm operator with the index of $A_{0}^{(2)}$ \cite{Kato1980}. 
\end{proof}

\subsubsection{Hamiltonian symmetry}
\label{sec:Ham}

We show in this section that the eigenvalues $\lambda$ of (\ref{ii}) are symmetrically placed with respect to the real and imaginary axes.
Lemma \ref{Asa} implies that the spectrum of $Q(\lambda)$ is symmetric with respect to $\mathbb{R}$, since $Q(\lambda)$ is invertible if and only if $Q^{*}(\lambda)$ is invertible.   
\begin{lemma}
The operators $A_{0}$, $A_{1}$ and $A_{2}$ have the properties $\overline{A_{0}u}=A_{0}\overline{u}$, $\overline{A_{1}u}=-A_{1}\overline{u}$ and $\overline{A_{2}u}=A_{2}\overline{u}$.
\label{Acomplex}
\end{lemma}
\noindent
\begin{proof} 
The lemma follows from the equalities
\[
	(\overline{A_{0}u},v)_{H^{1}(\mathbb{T}^{2})}=\overline{(A_{0}u,\overline{v})}_{H^{1}(\mathbb{T}^{2})}=a_{0}(\overline{u},v)=(A_{0}\overline{u},v)_{H^{1}(\mathbb{T}^{2})},
\]
\[
	(\overline{A_{1}u},v)_{H^{1}(\mathbb{T}^{2})}=\overline{(A_{1}u,\overline{v})}_{H^{1}(\mathbb{T}^{2})}=-a_{1}(\overline{u},v)=-(A_{1}\overline{u},v)_{H^{1}(\mathbb{T}^{2})},
\]
\[
	(\overline{A_{2}u},v)_{H^{1}(\mathbb{T}^{2})}=\overline{(A_{2}u,\overline{v})}_{H^{1}(\mathbb{T}^{2})}=a_{2}(\overline{u},v)=(A_{2}\overline{u},v)_{H^{1}(\mathbb{T}^{2})}.
\]
\end{proof}\\
\noindent
The theorem below states the Hamiltonian structure of the spectrum  $\sigma (Q(\lambda ))$. 

\begin{theorem}
The spectrum of the quadratic operator pencil $Q(\lambda)$ consists of eigenvalues of finite multiplicity that have the Hamiltonian structure $(\lambda,-\lambda,\overline{\lambda},-\overline{\lambda})$. 
\label{Hamiltonian}
\end{theorem}
\begin{proof} 
Theorem (\ref{Fredholm}) states that $Q(\lambda )$ is a Fredholm operator, which implies that the spectrum $\sigma (Q(\lambda ))$ consists of eigenvalues of finite multiplicity. Lemma \ref{Acomplex} gives the relation
\begin{equation}	\overline{Q(\lambda)u}=\overline{A_{0}u}+\overline{\lambda}\overline{A_{1}u}+\overline{\lambda}^{2}\overline{A_{2}u}=A_{0}\overline{u}-\overline{\lambda}A_{1}\overline{u}+\overline{\lambda}^{2}A_{2}\overline{u}=Q(-\overline{\lambda})\overline{u},
\end{equation}	
which implies that $-\overline{\lambda}$ is an eigenvalue whenever $\lambda$ is an eigenvalue. The two relations
\begin{equation} 
	Q^{*}(-\lambda)=A_{0}-\overline{\lambda}A_{1}+\overline{\lambda}^{2}A_{2}=Q(-\overline{\lambda}),
\end{equation}
\begin{equation}
	Q^{*}(\lambda)=A_{0}+\overline{\lambda}A_{1}+\overline{\lambda}^{2}A_{2}=Q(\overline{\lambda}),
\end{equation}
follows from Lemma \ref{Asa}.
\end{proof}
 
\subsection{Band gaps}

Assume that the periodic permitivity $\epsilon$ is smooth in $\RR$. If the equation
\begin{equation}
	Lv=0,\quad L=\Delta +\omega^{2} \epsilon(x,\omega)
	\label{eq:bas2}
\end{equation} 
has a bounded solution $|v(x)|\leq C$, then (\ref{eq:bas2}) also has a Bloch solution with a real quasimomentum $k$. The existence of a Bloch solution with $k\in\RR$ implies that the operator $L$ is not invertible in $L_{2}(\mathbb{R}^{2})$, see Chapter 4 in \cite{Kuchment1993}. We will use the above results to define a band-gap in terms of the Bloch variety.

Assume that the parameter $\omega$ in $Q(\lambda)$, that corresponds to the time frequency, is real. The complex Bloch variety of the pencil $Q(\lambda)$ is the set
\begin{equation}
	B(Q)=\{(k,\omega)\in\mathbb{C}^{2}\times\mathbb{R}\,|\, Q(\lambda)u=0\,\text{has a non-zero solution}\},
	\label{BVK}
\end{equation}
and the real Bloch variety of the pencil $Q(\lambda)$ is the intersection
\begin{equation}
	B_{\mrm{real}}(Q)=B(Q) \cap \mathbb{R}^{3}.
	\label{realBVK}
\end{equation}
We state below the condition for a band-gap in terms of the real Bloch variety.
\begin{definition}
The frequency $\omega$ is in a band gap if the projection of the real Bloch variety onto the $\omega$-axis is empty. 
\label{def:gap}
\end{definition}

Solutions of the form $(k,\omega)\in\mathbb{R}^{2}\times\mathbb{R}$ exist if the frequency $\omega$ is not in a band gap. That is, bounded Bloch waves exist for the given frequency $\omega$ and we say that $\omega$ belongs to the stability zone.  A gap in the spectrum of $Q(\lambda)$ is also called a polarization gap since we only consider the polarization $E(x)=(0,0,u(x))$. Notice that we let $k$ be real in the frequency independent case and defined a band gap in (\ref{spA}) as an empty intersection of two successive bands
\begin{equation}
	[\min_{k}\omega_{n}^{2}(k),\max_{k}\omega_{n}^{2}(k)]\cap  [\min_{k}\omega_{n+1}^{2}(k),\max_{k}\omega_{n+1}^{2}(k)]=\{\}.
\end{equation}
We also used that the spectrum alternatively can be described in terms of the projection of the real Bloch variety (\ref{realBV}) onto the $\omega$-axis. 

The dual (or reciprocal) lattice to $\Gamma=2\pi\mathbb{Z}^{2}$ is
\begin{equation}
	\Gamma^{*}=\{q\in\mathbb{R}^{2}\,|\, \gamma\cdot q\in 2\pi\mathbb{Z}, \forall\gamma\in\Gamma\}
\end{equation}
We define the fundamental domain (the Brillouin zone) of the dual lattice $\Gamma^{*}=\mathbb{Z}^{2}$ as the set
\begin{equation}
\Omega^{*}=\left (-\frac{1}{2},\frac{1}{2}\right ]^{2}.
\end{equation} 
The real part of all eigenvalues $\lambda$ to (\ref{classA}) is $\Omega^{*}$- periodic, since the Bloch solutions (\ref{Bloch}) are periodic in the real part of $k=\hat{k}\lambda$. We write the complex spectral parameter $\lambda$ in the form
\begin{equation}
	\lambda=\lambda_{\mrm{r}}+\lambda_{\mrm{i}}\iu,
\end{equation}
where $\lambda_{\mrm{r}}$ and $\lambda_{\mrm{i}}$ are real numbers. The set
\begin{equation}
	B_{\Omega^{*}}=\{(\lambda_{\mrm{r}}\hat{k},\lambda_{\mrm{i}}\hat{k},\omega)\in \Omega^{*}\times\mathbb{R}^{2}\times\mathbb{R}\,|\, Q(\lambda)u=0\, \text{has a non-zero solution}\},
\end{equation}
is a subset of the Bloch variety (\ref{BVK}). This set corresponds to solutions with the real part of the wave vector $k=\hat{k} \lambda$ in the Brillouin zone $\Omega^{*}$ and a real time frequency $\omega$. The frequency $\omega$ is in a band-gap if $\lambda_{\mrm{i}}\neq 0$ for all eigenvalues of $Q(\lambda)$, with $\hat{k}\lambda_{\mrm{r}}\in \Omega^{*}$.

When $\omega=0$ and $\hat{k}=(1,0)$, the lemma below shows that the eigenvalues can be calculated explicitly.
\begin{lemma}  
Let $\omega=0$ and $\hat{k}=(1,0)$. Then all eigenvalues $\lambda\in\mathbb{C}$ can be written of the form
\begin{equation}
	\lambda=\mathbb{Z}+\iu\mathbb{Z}.
\end{equation}
\label{nollan}
\end{lemma}
\begin{proof} 
Let $v(x)=\hat{v}_{m}\eu^{\iu m\cdot x}$, $m\in\mathbb{Z}^{2}$ and represent $u\in H^{1}(\mathbb{T}^{2})$ with its Fourier series
\begin{equation}
	u(x)=\sum_{n\in\mathbb{Z}^{2}}\hat{u}_{n}\eu^{\spnx}.
\end{equation}
The equation $(Q(\lambda)u,v)_{H^{1}(\mathbb{T}^{2})}=0$ gives
\begin{equation}
	 (m+\lambda\hat{k})\cdot (m+\lambda\hat{k})\hat{u}_{m}\hat{v}^{*}_{m}=0.
\end{equation}
The coefficients are zero whenever $\hat{u}_{m}\neq 0$ which implies
\begin{equation} 
  \left\{
    \begin{aligned}
			|m|^{2}+2\lambda_{\mrm{r}}\hat{k}\cdot m+\lambda_{\mrm{r}}^{2}-\lambda_{\mrm{i}}^{2}&= 0,\\
			\lambda_{\mrm{i}}\hat{k}\cdot m+\lambda_{\mrm{r}}\lambda_{\mrm{i}} &= 0.
    \end{aligned}\right. 
    \label{int}
\end{equation}  
Let $\hat{k}=(1,0)$ and assume $\lambda_{\mrm{i}}=0$. Then follows $\lambda=\lambda_{\mrm{r}}=-m_{1}$ and $m_{2}=0$. The secound equation in  (\ref{int}) is satisfied for $\lambda_{\mrm{r}}=-m_{1}$ and  $\lambda_{\mrm{i}}=\pm |m_{2}|$ follows from the first equation in (\ref{int}).
\end{proof}

Notice that eigenvalues $\lambda$ in the lemma above is real or purely imaginary if $m_{2}=0$. This case corresponds to a permittivity $\epsilon$ that only dependence on one coordinate. We illustrate the lemma in Section \ref{NumExp}.

The frequency $\omega=0$ is an eigenvalue of $A_{0}= A_{0}(\omega)$ and the corresponding eigenvector is constant. An eigenvalue $\omega$ of the operator-valued function  $A_{0}(\omega)$ corresponds to a solution of the classical eigenvalue problem
\begin{equation}
	-\Delta u=\omega^{2}\epsilon (x,\omega)u,
\end{equation}
where $u$ and $\epsilon$ are $\Gamma$-periodic. Assume that $(\tilde{u},\tilde{\omega})$ is a solution of   $A_{0}(\omega)u=0$. The pencil $Q(\lambda)$ can then be written of the form
\begin{equation}
	Q(\lambda)\tilde{u}=\lambda (A_{1}\tilde{u}+\lambda A_{2}\tilde{u}).
	\label{redK}
\end{equation}
Zero is always an eigenvalue of the pencil and a nonzero eigenvalue $\lambda$ is a solution of 
\begin{equation}
	 A_{1}\tilde{u}+\lambda A_{2}\tilde{u}=0,
\end{equation}
where $\tilde{u}\in N(A_{0}(\tilde{\omega}))$. All eigenvalues of (\ref{redK}) are real and the number of eigenvalues at $\omega=\tilde{\omega}$ depends on the dimension of the null space $N(A_{0}(\tilde{\omega}))$.

\section{Numerical treatment}

\subsection{Discretizing the problem}
In this section we discuss the numerical computation of approximate eigenvalues of the operator pencil $Q(\lambda)$. According to Theorem \ref{Hamiltonian} the spectrum of $Q(\lambda)$ has a Hamiltonian structure. A numerical method should preserve the symmetry of the spectrum. We will show that this can be achieved, if we discretize the weak formulation of the problem, which was given in
\eqref{ii}, by a conforming Galerkin ansatz.

Let $V_N$ be an $N$--dimensional linear subspace of $H^1(\mathbb{T}^2)$, and let $\{ \phi_1,\dotsc, \phi_N \}$ be a basis of $V_N$. The subspace $V_N$ can be constructed, for example, using a finite element ansatz. We seek $\lambda_N \in \mathbb{C}$ and
$u_N \in V_N$, such that
\begin{equation}
\lambda_N^2 a_2(u_N,v_N) + \lambda_N a_1(u_N,v_N) + a_0(u_N,v_N) = 0,
\end{equation}
for all $v_N \in V_N$. The approximate eigenfunctions are of the form
\begin{equation}
u_N = \sum_{j = 1}^N \alpha_j \phi_j,
\end{equation}
where the coefficients $\alpha_j$ are given by the solution of the quadratic matrix eigenvalue problem
\begin{equation}
\label{eq:quad-mat-qep}
\lambda^2_N \mat{A}_2 \mat{u} + \lambda_N \mat{A}_1 \mat{u} + \mat{A}_0 \mat{u} = 0,
\end{equation}
with $\mat{u} = (\alpha_1,\dotsc,\alpha_N)^\mathsf{T}$. The $N \times N$ matrices $\mat{A}_2$, $\mat{A}_1$ and 
$\mat{A}_0$ are given by
\begin{equation}
\begin{aligned}
(\mat{A}_0)_{mn} &= \int_{\mathbb{T}^2}\nabla \phi_n\cdot\nabla \phi_m-
\omega^2 \epsilon(\omega) \phi_n \phi_m\, \diff,\\
(\mat{A}_1)_{mn} &= 2 \iu \int_{\mathbb{T}^2} \phi_n\hat{k}\cdot\nabla \phi_m \diff,\\
(\mat{A}_2)_{mn} &= \int_{\mathbb{T}^2} \phi_n \phi_m\, \diff,\qquad &m, n&=1,\dotsc,N.
\end{aligned}
\end{equation}
Since $V_N$ is a subspace of $H^1(\mathbb{T}^2)$, we observe that the matrices in \eqref{eq:quad-mat-qep} have the same properties as the operators in the operator pencil (\ref{pencil}):
\begin{equation}
\mat{A}_2 = \mat{A}_2^\mathsf{H} > 0,\quad \mat{A}_1 = \mat{A}_1^\mathsf{H},\quad \mat{A}_0 = \mat{A}_0^\mathsf{H}.
\end{equation}
Note that the matrices $\mat{A}_2$ and $\mat{A}_0$ are real, while the matrix $\mat{A}_1$ is purely imaginary. Since we want to avoid complex matrix arithmetic in actual computations, we transform the problem as follows:
Letting $\lambda_N = -\iu\mu$, $\mat{K} = -\mat{A}_0$, $\mat{G} = \iu \mat{A}_1$ and $\mat{M} = \mat{A}_2$ 
the quadratic matrix eigenvalue problem can be rewritten as
\begin{equation}
\label{eq:quad-mat-gyroep}
\mu^2 \mat{M} \mat{u} + \mu \mat{G} \mat{u} + \mat{K} \mat{u} = 0.
\end{equation}
The $N \times N$ matrices in \eqref{eq:quad-mat-gyroep} are real and satisfy
\begin{equation}
\mat{M} = \mat{M}^\mathsf{T} > 0,\quad \mat{G} = -\mat{G}^\mathsf{T},\quad \mat{K} = \mat{K}^\mathsf{T}.
\end{equation}
The quadratic matrix eigenvalue problem in \eqref{eq:quad-mat-gyroep} is therefore called a \emph{gyroscopic matrix eigenvalue problem}. This name refers to the fact, that quadratic matrix eigenvalue problems of this form typically occur in the analysis of undamped gyroscopic systems \cite{tisseur2001}. It is easy to show, that the set of eigenvalues of a gyroscopic eigenvalue problem has a Hamiltonian structure: Suppose that $\mu$ is an eigenvalue of \eqref{eq:quad-mat-gyroep}. Then, there exists an eigenvector $\mat{u} \in \mathbb{C}^N$, $\mat{u} \neq 0$, such that
\begin{equation}
\mu^2 \mat{M} \mat{u} + \mu \mat{G} \mat{u} + \mat{K} \mat{u} = 0.
\end{equation}
Due to the structure of the matrices in \eqref{eq:quad-mat-gyroep}, we have
\begin{equation}
\label{eq:mat-hamiltonian}
\begin{aligned}
\overline{\mu}^2 \mat{M} \overline{\mat{u}} + \overline{\mu} \mat{G} \overline{\mat{u}} + \mat{K} \overline{\mat{u}}
&=\overline{\mu^2 \mat{M} \mat{u} + \mu \mat{G} \mat{u} + \mat{K} \mat{u}}
= 0,\\
(-\mu)^2 \mat{M} \mat{u} - \mu \mat{G} \mat{u} + \mat{K} \mat{u}
&=(\mu^2 \mat{M} \mat{u} - \mu \mat{G} \mat{u} + \mat{K} \mat{u})^\mathsf{T}
= 0,\\
(-\overline{\mu})^2 \mat{M} \overline{\mat{u}} - \overline{\mu} \mat{G} \overline{\mat{u}} + \mat{K} \overline{\mat{u}}
&=(\mu^2 \mat{M} \mat{u} - \mu \mat{G} \mat{u} + \mat{K} \mat{u})^\mathsf{H}
= 0.
\end{aligned}
\end{equation}
From \eqref{eq:mat-hamiltonian} follows that every element of the set $\{ \mu, \overline{\mu}, -\mu, -\overline{\mu} \}$ is an eigenvalue. Hence, the discrete problem will preserve the Hamiltonian structure of the continuous problem. This statement remains true as long as the discrete space $V_N$ is a subspace of $H^1(\mathbb{T}^2)$. In the context of finite element methods this means that the finite element basis functions satisfy periodic boundary conditions.

\subsection{Linearization}

The numerical solution of quadratic and especially of gyroscopic matrix eigenvalue problems is still an active field of research \cite{lancaster2003, ruhe2000, ferng1999}. A comprehensive overview of numerical algorithms that can be applied to quadratic matrix eigenvalue problems can be found in \cite{tisseur2001}.

A common way to treat quadratic matrix eigenvalue problems is to transform the quadratic problem to an equivalent generalized eigenvalue problem of double dimension \cite{gohberg+lancaster+rodman1982,tisseur2001}. The gyroscopic matrix eigenvalue problem in \eqref{eq:quad-mat-gyroep} can be linearized as
\begin{equation}
\label{eq:lin-gep}
\blk{A} \blk{x} = \mu \blk{B} \blk{x},
\end{equation}
where the $2N \times 2N$ block matrices $\mat{A}$ and $\mat{B}$ are given by
\begin{equation}
\blk{A} =
\begin{pmatrix}
\phantom{-}\mat{0} & \phantom{-}\mat{I} \\
-\mat{K} & -\mat{G}
\end{pmatrix},\quad
\blk{B} =
\begin{pmatrix}
\mat{I} & \mat{0} \\
\mat{0} & \mat{M}
\end{pmatrix},\quad
\blk{x} =
\begin{pmatrix}
\mat{u} \\ \mu \mat{u}
\end{pmatrix}.
\end{equation}
Since $\mat{M}$ is positive definite and $\mat{I}$ is the identity matrix, the block matrix $\blk{B}$ is also positive definite. The computation of the  eigenvalues $\mu$ can be done, for example, with the QZ-algorithm \cite{tisseur2001}. 

It turns out, however, that this algorithm is not well-suited for the compuation of photonic band gaps, mainly for two reasons: First, the matrices are typcially large and sparse. This is due to the fact, that they are generated by a finite element ansatz. The QZ-algorithm will destroy the sparsity in general and therefore lead to excessive storage requirements. Second, only a few eigenvalues need to be computed in order to determine the photonic band structure. For these two reasons Krylov space methods, such as an implicitely restarted Arnoldi process (IRA) or an implicitely restarted Lanczos process (IRL), are preferable
\cite{saad1992}. Using shift-and-invert strategies these algorithms can be used to compute a small number of eigenvalues $\mu_1, \dotsc, \mu_M$ that are closest to a prescribed shift parameter $\mu_0$. Each computed eigenvalue $\mu_\ell$ determines up to four  eigenvalues of the gyroscopic matrix eigenvalue problem, namely $\{ \mu_\ell, \overline{\mu}_\ell, -\mu_\ell, -\overline{\mu}_\ell \}$.

\subsection{The SHIRA algorithm}
A drawback of standard Krylov space methods is that they do not take into account the Hamiltonian structure of the set of eigenvalues of the gyroscopic matrix eigenvalue problem. Suppose that $\mu \in \mathbb{C}$ is an eigenvalue for \eqref{eq:quad-mat-gyroep} that is close to the real axis. A Krylov space method may find eigenvalues $\mu_1, \dotsc, \mu_M$, amoung which are approximations to $\mu_\ell$ and $\mu_{\ell'}$ to $\mu$ and to $\overline{\mu}$, respectively. Numerical round-off errors could however lead to
$\mu_\ell \neq \overline{\mu}_{\ell'}$. Thus, one would get a wrong extra quadruple of eigenvalues.

As a remedy, one might consider to use structure-preserving algorithms \cite{hwang2003}, like the SHIRA algorithm \cite{mehrmann2001}. SHIRA stands for \emph{skew-Hamiltonian implicitely restarted Arnoldi process}. As the name suggests, it is a modified Arnoldi process
that preserves the Hamiltonian symmetry of the eigenvalues of a skew-Hamiltonian matrix. The algorithm can be used to compute a small number of eigenvalues for a large and sparse gyroscopic eigenvalue problem. Below we summarize the SHIRA algorithm and refer to 
\cite{mehrmann2001} for details.

The the gyroscopic matrix eigenvalue problem in \eqref{eq:quad-mat-gyroep} is linearized as
\begin{equation}
\label{eq:lin-shep}
\blk{H} \blk{y} = \mu \blk{S} \blk{y},
\end{equation}
where the $2N \times 2N$ block matrices $\blk{H}$ and $\blk{S}$ are given by
\begin{equation}
\mat{H} =
\begin{pmatrix}
\mat{0} & -\mat{K} \\
\mat{M} & \mat{0}
\end{pmatrix},\quad
\mat{S} =
\begin{pmatrix}
\mat{M} & \mat{G} \\
\mat{0} & \mat{M}
\end{pmatrix},\quad
\mat{y} =
\begin{pmatrix}
\mu \mat{u} \\ \mat{u}
\end{pmatrix}.
\end{equation}
The matrix $\blk{H}$ turns out to be a \emph{Hamiltonian} matrix, i.e., it satisfies $(\blk{H}\blk{J})^\mathsf{T} = \blk{H}\blk{J}$, where the $2N \times 2N$ block matrix $\blk{J}$ is the imaginary unit matrix
\begin{equation}
\mat{J} = \begin{pmatrix}
\phantom{-}\mat{0} & \mat{I} \\ -\mat{I} & \mat{0}
\end{pmatrix}.
\end{equation}
The block matrix $\blk{S}$ is \emph{skew-Hamiltonian}, i.e., it satisfies $(\blk{S}\blk{J})^\mathsf{T} = -\blk{S}\blk{J}$. The factorization
\begin{equation}
\label{eq:fact}
\blk{S} = 
\begin{pmatrix}
\mat{I} & \frac{1}{2}\mat{G} \\
\mat{0} & \mat{M}
\end{pmatrix}
\begin{pmatrix}
\mat{M} & \frac{1}{2}\mat{G} \\
\mat{0} & \mat{I}
\end{pmatrix}
\end{equation}
allows us to rewrite \eqref{eq:lin-shep} as
\begin{equation}
\label{eq:linearization}
\blk{W} \blk{y} = \mu \blk{y}
\end{equation}
where $\blk{W}$ is the $2N \times 2N$ block matrix
\begin{equation}
\mat{W} = \begin{pmatrix}
\mat{I} & \frac{1}{2}\mat{G} \\
\mat{0} & \mat{M}
\end{pmatrix}^{-1}
\begin{pmatrix}
\mat{0} & -\mat{K} \\
\mat{M} & \mat{0}
\end{pmatrix}
\begin{pmatrix}
\mat{M} & \frac{1}{2}\mat{G} \\
\mat{0} & \mat{I}
\end{pmatrix}^{-1}.
\end{equation}
The matrix $\mat{W}$ turns out to be Hamiltonian \cite{mehrmann2001}. Since the spectrum of such a matrix has a Hamiltonian symmetry, the
eigenvalues of the matrix
\begin{equation}
\blk{R} = (\blk{W} - \mu_0)^{-1}(\blk{W} + \mu_0)^{-1}(\blk{W} - \overline{\mu_0})^{-1}(\blk{W} + \overline{\mu_0})^{-1} 
\end{equation}
will have the same symmetry for every shift parameter $\mu_0 \in \mathbb{C}$. Moreover, the eigenvalues of $\blk{R}$ that are largest in modulus correspond to eigenvalues of $\blk{W}$ that are closest to elements of the set $\{\mu_0,-\mu_0,\overline{\mu_0},-\overline{\mu_0}\}$. Therefore, an implicitely restarted Arnoldi process will in general converge to these eigenvalues. It should be noted that the application of $\blk{R}$ can be realized efficiently in actual computer programs.

It can be shown that the block matrix $\blk{R}$ is skew--Hamiltonian. In exact arithmetic, the Krylov spaces 
\begin{equation}
\mathcal{K}_n = \mathop{\mathrm{span}}\{ \blk{v}_0, \blk{R}\blk{v}_0, \blk{R}^2\blk{v}_0, \dotsc, \blk{R}^{n-1}\blk{v}_0 \} 
\end{equation}
that are generated by $\blk{R}$ should therefore be \emph{isotropic}, i.e., they should satisfy $\mathcal{K}_n \perp \mat{J}\mathcal{K}_n$. Due to numerical round-off errors, however, this property is usually lost after some Krylov space
iterations, and this will eventually destroy the Hamiltonian symmetry of the computed eigenvalues. In order to cope with this problem, SHIRA introduces an additional re-orthogonolization step, which ensures that the Krylov spaces remain isotropic up to machine precision.
Because of this, SHIRA always converges to complex conjugate pairs of eigenvalues in the right hand complex halfplane. The algorithm  thus prevents the occurrence of wrong extra eigenvalue quadruples.

The SHIRA algorithm has already been successfully used to compute the corner singularities in elasticity problems \cite{Apel2002}. In our computations we found that in general it needs less implicit restarts than a standard implicitely restarted Arnoldi process. However, in some cases, especially when the relevant eigenvalues where close to each other, we could not obtain convergence with the SHIRA algorithm. In these cases we resorted to the linearization in \eqref{eq:lin-gep} and used an IRA algorithm as provided by 
the ARPACK package.

\begin{figure}
\begin{center}
\includegraphics[width=5.5cm]{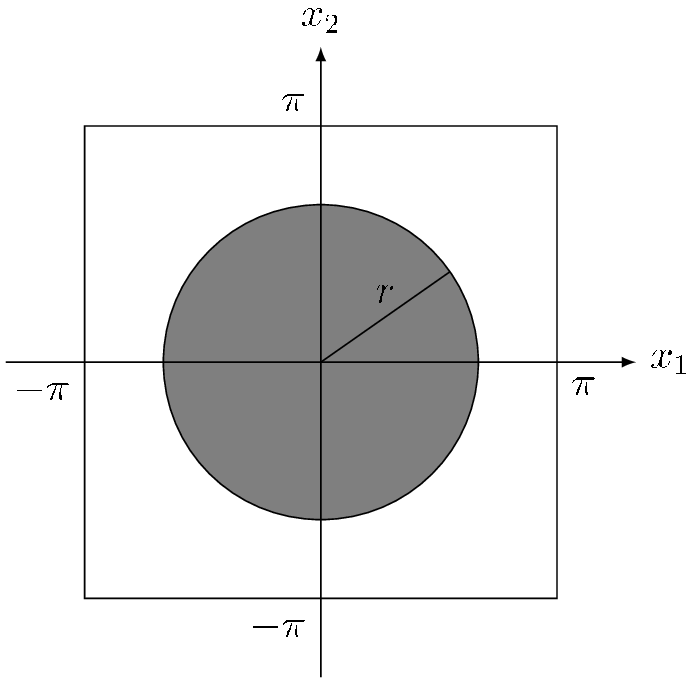}
\hspace{0.5cm}
\includegraphics[width=5.5cm]{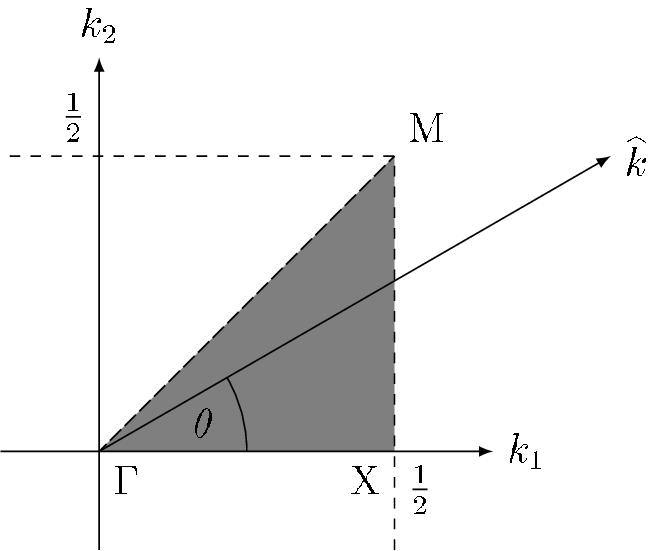}
\end{center}
\caption{The Wigner--Seitz cell of the photonic crystal (left) and the irreducible part of the Brillouin zone (right).}
\label{fig:cell_zone}
\end{figure}

\section{Numerical examples}
\label{NumExp}

To compute the band gaps, we choose frequencies $\omega$ in a frequency range $[\omega_a,\omega_b]$ and let
\begin{equation}
\widehat{k} = \begin{pmatrix} \cos \theta \\ \sin \theta \end{pmatrix}.
\end{equation}
One can show that the band structure can be calculated reducing the angles to $0\leq\theta\leq \pi/4$, which is called the irreducible Brillouin zone \cite{Kittel1986}. The operator pencil $Q(\lambda)$ is discretized with quadratic elements, which  results in a quadratic matrix eigenvalue problem on the form (\ref{eq:quad-mat-gyroep}). Solving these quadratic matrix eigenvalue problems, we obtain for each frequency
$\omega$ and each angle $\theta$ a set of eigenvalues $\Lambda(\omega,\theta)$. Eigenvalues $\lambda \in \Lambda(\omega,\theta)$ that satisfy $\lvert\lambda\rvert \leq 1/\cos \theta$, correspond to quasi-momentum vectors $k = \lambda \widehat{k}$
that lie in the first Brillouin zone. A frequency $\omega$ is a band gap frequency if the imaginary parts of all eigenvalues in $\Lambda(\omega,\theta)$ do not vanisch for all $\theta$.  

We compute band gaps when the periodic structure is a rectangular array of dielectric cylinders in air. The relative diameter of the cylinders with respect to the lattice constant is $0.75$. Below we consider one frequency independent model and one frequency dependent model. The mesh width is $h = 0.025$ in both cases, which leads to quadratic matrix eigenvalue problems with approximately $10^5$ unknowns. The method was implemented in \emph{Matlab} and the computations were performed on a standard laptop.

\begin{figure}
\begin{center}
\includegraphics[width=6cm]{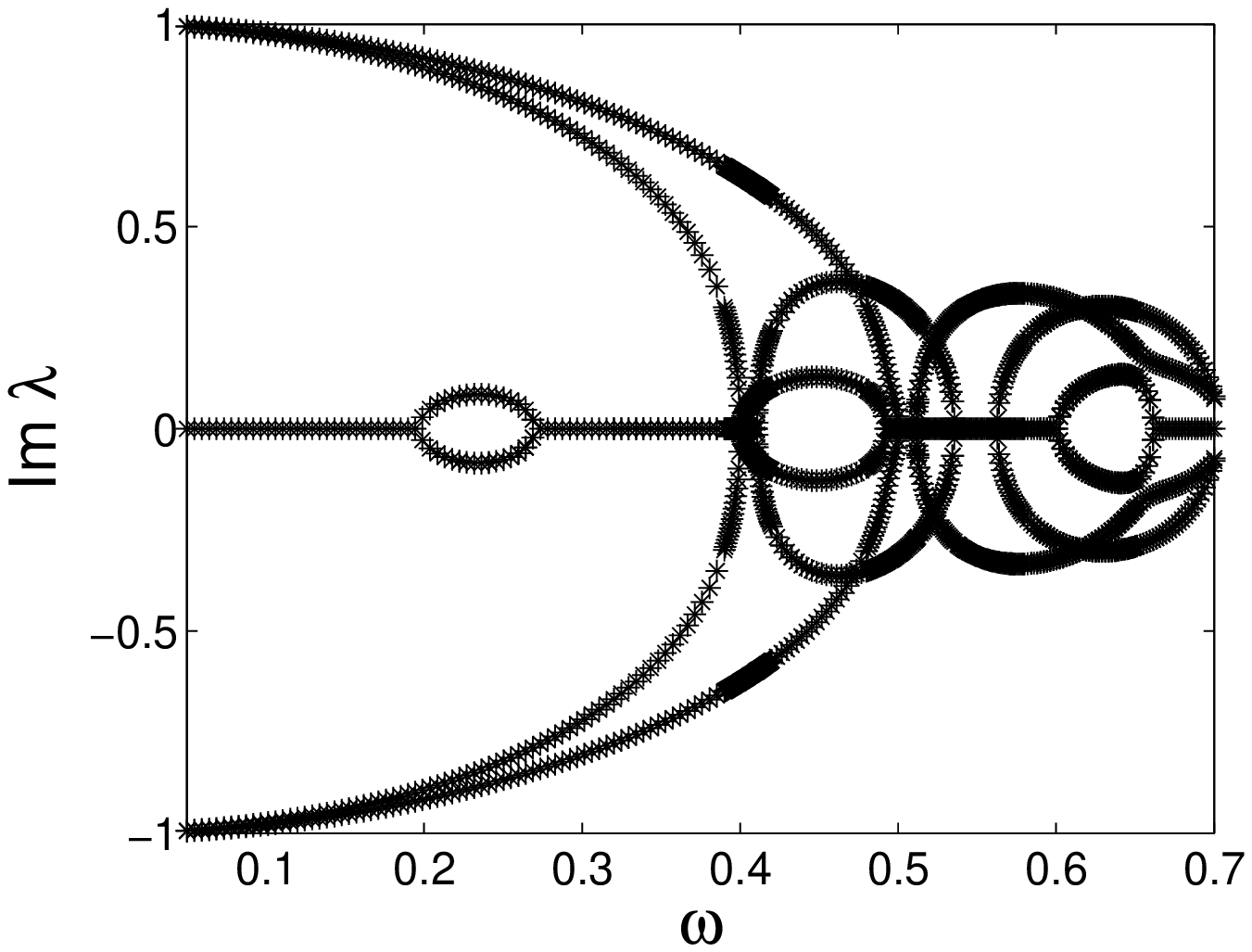}
\includegraphics[width=6cm]{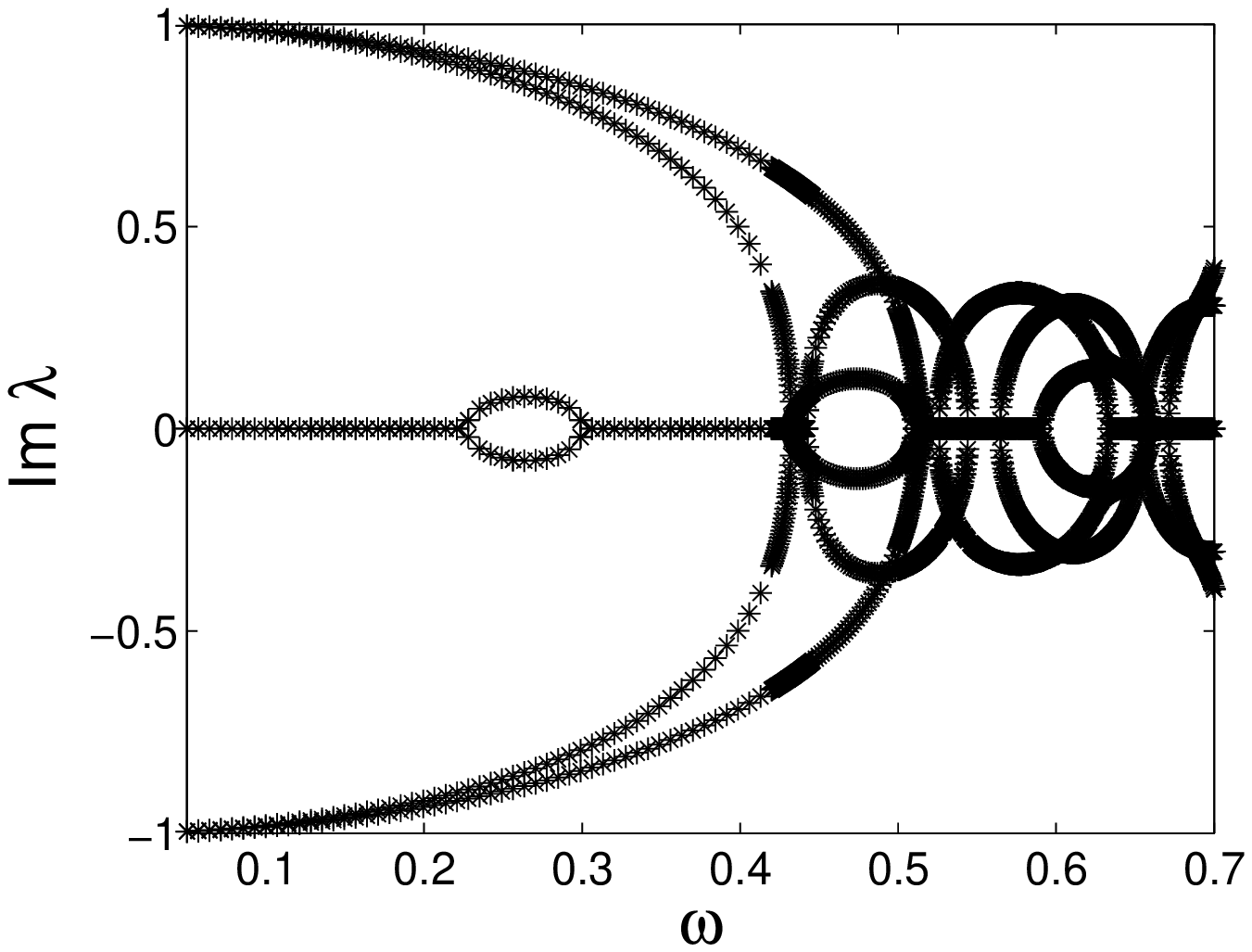}
\end{center}
\caption{
Imaginary part of $\lambda$ as a function of the frequency $\omega$ when $\hat{k}=(1,0)$. Left: Material parameters $\epsilon_{1}=1$ and $\epsilon_{2}=8.9$. Right: $\epsilon_{1}=1$ and $\epsilon_{2}$ according to the model (\ref{material}).}
\label{fig:imag}
\end{figure}

\subsection{A frequency independent model}

\begin{table}
\begin{tabular}{|p{0.8cm}|p{1cm}|p{2.9cm}|p{2.9cm}|p{2.9cm}|}
\hline
$h/2\pi$ &  DOFs     & First gap             & Second gap            & Third gap            \\
\hline
$0.050$  &  $  2608$  &  $(.24719,\,.27015)$  &  $(.41064,\,.45632)$  &  $(.61757,\,.66173)$  \\
$0.040$  &  $  4124$  &  $(.24710,\,.27016)$  &  $(.41064,\,.45631)$  &  $(.61755,\,.66171)$  \\
$0.030$  &  $  7412$  &  $(.24718,\,.27015)$  &  $(.41063,\,.45631)$  &  $(.61754,\,.66171)$  \\
$0.025$  &  $ 10872$  &  $(.24711,\,.27015)$  &  $(.41029,\,.45650)$  &  $(.61754,\,.66171)$  \\
\hline
\end{tabular}
\caption{Computed band gaps for the frequency independent model (Dobson).}
\label{bandgap1}
\end{table}

As a first illustration, we assume that the permittivity of the cylinders is $\epsilon_2 = 8.9$, and that the permittivity of air is $\epsilon_1 = 1$ for all frequencies under consideration. This model was numerically investigated in \cite{Dobson1999}. Since the permittivity is frequency independent, the band gaps can be calculated from the shifted problem (\ref{shifted2}). Dobson  \cite{Dobson1999} used this formulation and calculated $\omega(k)$ where $k$ is on the boundary of the irreducible Brillouin zone; See Figure \ref{fig:cell_zone}. It is know that a band gap generically opens on the boundary, but it is possible to construct a material function that gives a band gap opening inside the Brillouin zone \cite{Harrison+Kuchment+Sobolev+Winn2007}.

With the approach in this paper, a band gap corresponds to a frequency $\omega$, where the imaginary part of all eigenvalues $\lambda$ are non-zero. Figure \ref{fig:imag} shows the imaginary part of the eigenvalues $\lambda (\omega,\hat{k})$ with smallest imaginary part when $\hat{k}=(1,0)$ and $\omega\in [0,0.7]$. Notice that $\Im \lambda$ is close to $\pm 1$ for $\omega$ close to zero. We know from Lemma \ref{nollan} that all solutions have $\Im \lambda\in \Z$ when $\omega=0$. To compute band gaps, $\hat{k}$ is varied over the irreducible Brillouin zone with small steps in $\theta$ and $\omega\in [0,0.7]$. The band gap frequencies in Table \ref{bandgap1} match well with Dobson's results \cite{Dobson1999}.

\begin{table}[t]
\begin{tabular}{|p{0.8cm}|p{1cm}|p{2.9cm}|p{2.9cm}|p{2.9cm}|}
\hline
$h/2\pi$ &  DOFs     & First gap             & Second gap            & Third gap            \\
\hline
$0.050$  &  $  2608$  &  $(.28231,\,.30031)$  &  $(.44251,\,.47728)$  &  $(.60209,\,.63320)$  \\
$0.040$  &  $  4124$  &  $(.28230,\,.30031)$  &  $(.44251,\,.47728)$  &  $(.60208,\,.63318)$  \\
$0.030$  &  $  7412$  &  $(.28231,\,.30031)$  &  $(.44250,\,.47727)$  &  $(.60207,\,.63318)$  \\
$0.025$  &  $ 10872$  &  $(.28228,\,.30031)$  &  $(.44238,\,.47784)$  &  $(.60207,\,.63318)$  \\
\hline
\end{tabular}
\caption{Computed band gaps for the frequency dependent model (\ref{material}).}
\label{bandgap2}
\end{table}  

\subsection{A frequency dependent model}

\begin{figure}
\begin{center}
\includegraphics[width=12cm]{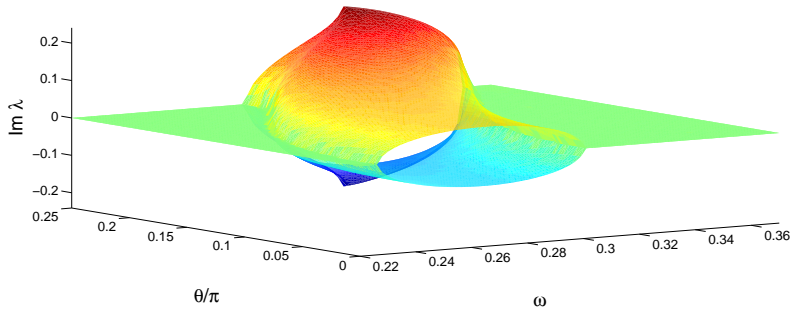}
\caption{
The first band gap tube $\Im \lambda (\omega,\theta)$ for the material model (\ref{material}).} 
\label{fig:Gap1a}
\end{center}
\end{figure}

We assume that the matrix material (air) has the permittivity $\epsilon_{1}=1$ for all frequencies under consideration. The permittivity of the dielectric cylinders is generally a complicated function of the frequency. It is common in practice to use least-squares minimization to determine a rational function approximation from frequency-domain data. The choice of the material model is not important for the algorithm, since the frequency dependence is only a parameter.  We will consider the simple model
\begin{equation}
	\epsilon_{2}(\omega)=a+\frac{b}{c-\omega^{2}},
	\label{material}
\end{equation}
where $a=1$, $b=5.34$, and $c=1$.  The constants in (\ref{material}) are chosen such that the mean value $(\epsilon (0)+\epsilon (0.7))/2$ is approximately $8.9$. We assume that this model is valid for all frequencies in the range $\omega\in [0,0.7]$. This assumption make it possible to compare the frequency dependent model with the calculations for the frequency independent case above. Figure \ref{fig:imag} shows the imaginary part of the eigenvalues $\lambda (\omega,\hat{k})$ with smallest imaginary part when $\hat{k}=(1,0)$ and $\omega\in [0,0.7]$. A frequency $\omega$ is in a band gap when all eigenvalues have a non-zero imaginary part. The lowest band gap in Table \ref{bandgap2} corresponds in Figure \ref{fig:Gap1a} to the tube. The band gap can alternatively be illustrated with spectral surfaces $\omega (\theta,\lambda)$, where $\lambda$ is a real eigenvalue. The spectral surfaces in Figure \ref{fig:Gap1b} are computed with the presented method by sorting the computed eigenvalues $\lambda (\omega,\hat{k})$.
 
\begin{figure}
\begin{center}
\includegraphics[width=9cm]{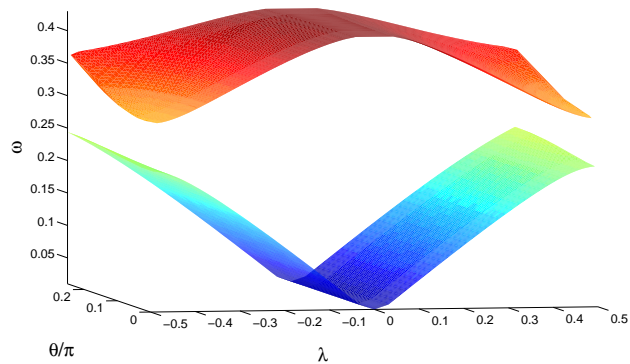}
\caption{
The spectral surfaces $\omega (\theta,\lambda)$ for the material model (\ref{material})} 
\label{fig:Gap1b}
\end{center}
\end{figure}

\section{Discussion and conclusions}

We have analyzed a method to compute Bloch waves in periodic and frequency dependent materials. The problem is formulated as a quadratic operator pencil in the quasimomentum $k$. This formulation is, in the case of frequency dependent materials, a significant simplification of usual approach, where the time frequency $\omega$ is calculated as a function of $k$. The pencil can be linearized since the non-linearity is polynomial. This is important, because many robust numerical algorithms for computing eigenvalues are only available for linear eigenvalue problems. The problem was discretized with finite elements and approximate eigenvalues was successfully computed with an implicitly restarted Arnoldi process.

\section{Acknowledgments}

The authors express their gratitude to Professor Kurt Busch, Professor Willy D\"{o}rfler, Thomas Gauss, Armin Lechleiter, Melanie Reimers and Arne Schneck for helpful discussions and acknowledges the support of the German Research Foundation (RTG 1294). 

\bibliography{total}
\bibliographystyle{siam}

DEPARTMENT OF MATHEMATICS, KARLSRUHE UNIVERSITY,\\ GERMANY\\\\
\noindent
\textsl{E-mail address}:\\
\noindent
\textbf{christian.engstroem@math.uni-karlsruhe.de},\\
\textbf{markus.richter@math.uni-karlsruhe.de} 
\end{document}